\documentclass[11pt]{article}
\usepackage{amsmath, amssymb}
\usepackage{epsfig}
\usepackage{subfigure}
\usepackage{graphicx}
\begin{document}

\title{
Occupancy distributions of homogeneous queueing systems under
opportunistic scheduling\thanks{An earlier version of this
manuscript appeared at the Information Theory and Applications
Workshop, UCSD, 2008.} }

\author{
Murat Alanyali and Maxim Dashouk\\
Department of Electrical and Computer Engineering \\
Boston University \\ }\date{} \maketitle \thispagestyle{empty}

\newtheorem{condition}{\indent \bf Condition}[section]
\newtheorem{property}{\indent \bf Property}[section]
\newtheorem{definition}{\indent \bf Definition}[section]
\newtheorem{conj}{\indent \bf Conjecture}[section]
\newtheorem{cor}{\indent \bf Corollary}[section]
\newtheorem{lemma}{\indent \bf Lemma}[section]
\newtheorem{claim}{\indent \bf Claim}[section]
\newtheorem{theorem}{\indent \bf Theorem}[section]
\newtheorem{prop}{\indent \bf Proposition}[section]
\newtheorem{remark}{\indent \bf Remark}[section]
\newtheorem{example}{\indent \bf Example}[section]

\newenvironment{proof}{\indent {\bf Proof.}}{\hfill$\bf \Box$\bigskip}
\newenvironment{proofl}{\indent {\bf Proof of Lemma }}{\hfill$\bf \Box$\bigskip}
\newenvironment{proofp}{\indent {\bf Proof of Proposition }}{\hfill$\bf \Box$\bigskip}
\newenvironment{prooft}{\indent {\bf Proof of Theorem }}{\hfill$\bf \Box$\bigskip}

\newcommand{\xt}[1][]{\mathbf{x}_{t #1}}
\newcommand{\bu}{\mathbf{u}}
\newcommand{\bb}{\mathbf{b}}
\newcommand{\bv}{\mathbf{v}}
\newcommand{\bx}{\mathbf{x}}
\newcommand{\bmm}{\mathbf{m}}
\newcommand{\be}{\mathbf{e}}
\newcommand{\bc}{\mathbf{c}}
\newcommand{\bm}{\mbox{\boldmath$\omega$}}
\newcommand{\bqa}{\begin{eqnarray}}
\newcommand{\eqa}{\end{eqnarray}}
\newcommand{\beq}{\begin{equation}}
\newcommand{\eeq}{\end{equation}}
\newcommand{\bqaa}{\begin{eqnarray*}}
\newcommand{\eqaa}{\end{eqnarray*}}

\renewcommand{\vec}[1]{\mbox{\boldmath$#1$}}

\begin{abstract}
We analyze opportunistic schemes for transmission
scheduling from one of $n$ homogeneous queues whose channel states
fluctuate independently. Considered schemes consist of the LCQ
policy, which transmits from a longest connected queue in the
entire system, and its low-complexity variants that transmit from
a longest queue within a randomly chosen subset of connected
queues. A Markovian model is studied where mean packet
transmission time is $n^{-1}$ and packet arrival rate is
\mbox{$\lambda<1$} per queue. Transient and equilibrium
distributions of queue occupancies are obtained in the limit as
the system size $n$ tends to infinity.
\end{abstract}


\newpage

\section{Introduction}

We analyze a queueing system that arises under opportunistic
scheduling of packet transmissions from a collection of queues
with time-varying service rates. The system of interest is
motivated by cellular data communications in which a single
transceiver serves multiple mobile stations through distinct
channels. Transmission scheduling has been well-studied in this
context under the guiding principle of opportunism, which broadly
refers to exploiting channel variations to maximize transmission
capacity in the long term. In this we paper consider two generic
opportunistic scheduling policies and obtain asymptotically exact
descriptions of the resulting queue length distributions in
symmetric systems of statistically identical queues.

Explicit analysis of queue lengths under opportunistic scheduling
is generally difficult due to model complexities and lack of
closed-form expressions. In related work Tassiulas and
Ephremides~\cite{EphTes} considered a queueing system under an
{\em on/off} channel model in which each queue is independently
either connected, and in turn it is eligible for service at a
standard rate, or disconnected and it cannot be serviced. It is
shown that transmitting from a longest connected queue stabilizes
queue lengths if that is at all feasible, and that it minimizes
occupancy of symmetric systems in which queues have identical load
and channel statistics. This policy is coined {\em LCQ}. Explicit
description of queue length distributions under LCQ is not
available, but several bounds for mean packet delay are obtained
in~\cite{Ganti,NeelyArl} for LCQ and some of its variants. In more
general models that admit multiple transmission rates and
simultaneous transmissions, {\em max-weight} scheduling policies
and their variations are shown in~\cite{Stolyar04,stoy} to
asymptotically minimize a range of occupancy measures along a
certain heavy-traffic limit. In the special case when one queue
can transmit at a time, max-weight transmits from a queue that
maximizes the product of instantaneous queue length and
transmission rate. Tails of queue length distributions under such
policies are studied in~\cite{Shakkottai08,Ying06} via large
deviations analysis.

Here we consider a system of $n$ queues under an on/off channel
model in which each queue is connected independently with
probability $q\in(0,1]$. A continuous-time Markovian model is
adopted where packet transmission rate is $n$ and packet arrival
rate is \mbox{$\lambda<1$} per queue. It can be seen that
$\lambda$ is also the load factor of the system; hence the
condition $\lambda<1$ is necessary to have positive-recurrent
queue lengths. We analyze this system for large values of the
system size $n$, under the LCQ scheduling policy and under its
low-complexity variant, namely $LCQ(d)$, that transmits from a
longest queue within $d\ge 1$ randomly selected connected queues.
It is apparent that LCQ($d$) is not particularly suitable for
non-symmetric systems, yet our goal here is to obtain a generic
evaluation of its underlying principle, which may be tailored to
specific circumstances.

We establish that as the system size $n$ increases equilibrium
distribution of queue occupancies under the LCQ policy converges
to the deterministic distribution centered at 0. Hence
asymptotically almost all queues are empty in equilibrium. The
number of queues with one packet is $\Theta(1)$ and the number of
queues with more than one packet is $o(1)$  as
$n\rightarrow\infty$. In particular maximum queue size tends to
one. The total number of packets in the system is therefore given
by the number of nonempty queues, and this number is shown to have
the same equilibrium distribution as the positive-recurrent
birth-death process with birth rate $\lambda$ and death rate
$1-(1-q)^j$ at state $j\in \mathbb{Z}_+$. Note that the latter
rate is equal to the probability of having at least one connected
queue within a given set of $j$ queues, and that the nature of the
total system occupancy may be anticipated once maximum queue size
is determined to be one. The obtained description leads to
asymptotic mean packet delay via Little's law as the rate of
packet arrivals to the system is readily seen to be $n\lambda$.

The analysis technique applied to LCQ can be extended, although
with excessive tediousness, to symmetric max-weight policies in
cases when each queue can be serviced independently at rate $nR$
for some random variable $R$. Above conclusions about LCQ offer
substantial insight about queue occupancies in that more general
setting. Namely if $R$ exceeds $\lambda$ with positive probability
(note that this condition is necessary for positive recurrence of
queue lengths), then stochastic coupling with a related LCQ system
yields that the maximum queue length in equilibrium tends to 1 as
$n\rightarrow\infty$. In~turn, equilibrium distribution of total
system occupancy should be expected to resemble that of a
birth-death process with birth rate $\lambda$ and death rate
$E[\max\{R_1,R_2,\cdots,R_j\}]$ at state $j$, where
$R_1,R_2,\cdots,R_j$ are independent copies of $R$.

We obtain the equilibrium distribution $\{p_k\}_{k=0}^\infty$ of
individual queue occupancy under the LCQ($d$) policy in the limit
as $n\rightarrow\infty$. Specifically $p_k=v^*_k-v^*_{k+1}$ where
$v^*_0=1$ and \[v^*_k=1-\sqrt[d]{1-\lambda v^*_{k-1}},~~~~~~k =
1,2,\cdots. \] This distribution is shown to have tails that decay
as $\Theta((\lambda/d)^k)$ as queue size $k\rightarrow\infty$.
Hence, in terms of tail occupancy probabilities, the choice
parameter $d$ has the equivalent effect of reducing the system
load by the same factor. Yet, for any fixed $d$, system occupancy
under LCQ($d$) is larger than that of LCQ by a factor of order
$n$. Numerical values of asymptotic mean packet delay under LCQ
and LCQ($d$) are illustrated in Figure~\ref{fig.delay}. We also
conclude that if $d$ is allowed to depend on $n$ so that
$d\rightarrow\infty$ but $d/n\rightarrow 0$ then order of the
alluded disparity reduces to $n/d$. This suggests that for
moderate values of $d$ and $n$ LCQ($d$) and LCQ may be expected to
have comparable packet delay.

The present analysis is based on approximating system dynamics via
asymptotically exact differential representations that amount to
functional laws of large numbers. Hence besides the mentioned
equilibrium properties the paper also describes transient behavior
of queue occupancies. The present analysis of LCQ($d$) is inspired
by the work of Vvedenskaya et al.~\cite{vveden} which concerns an
analogue of this policy that arises in routing and load balancing.
It should perhaps be noted that the choice parameter $d$ appears
to have a substantially more pronounced effect in the routing
context. Our conclusions about the LCQ policy require a more
elaborate technical approach. Here we apply a technique due to
Kurtz~\cite{Kurtz92} to obtain a suitable asymptotic description
of the system. Related applications of this technique can be found
in~\cite{MA03,HuntKurtz94,Zachary02}.

\begin{figure*}
\label{fig.delay}\centering \hspace*{-0.9cm}
\subfigure[LCQ]{\includegraphics[width=3.8in]{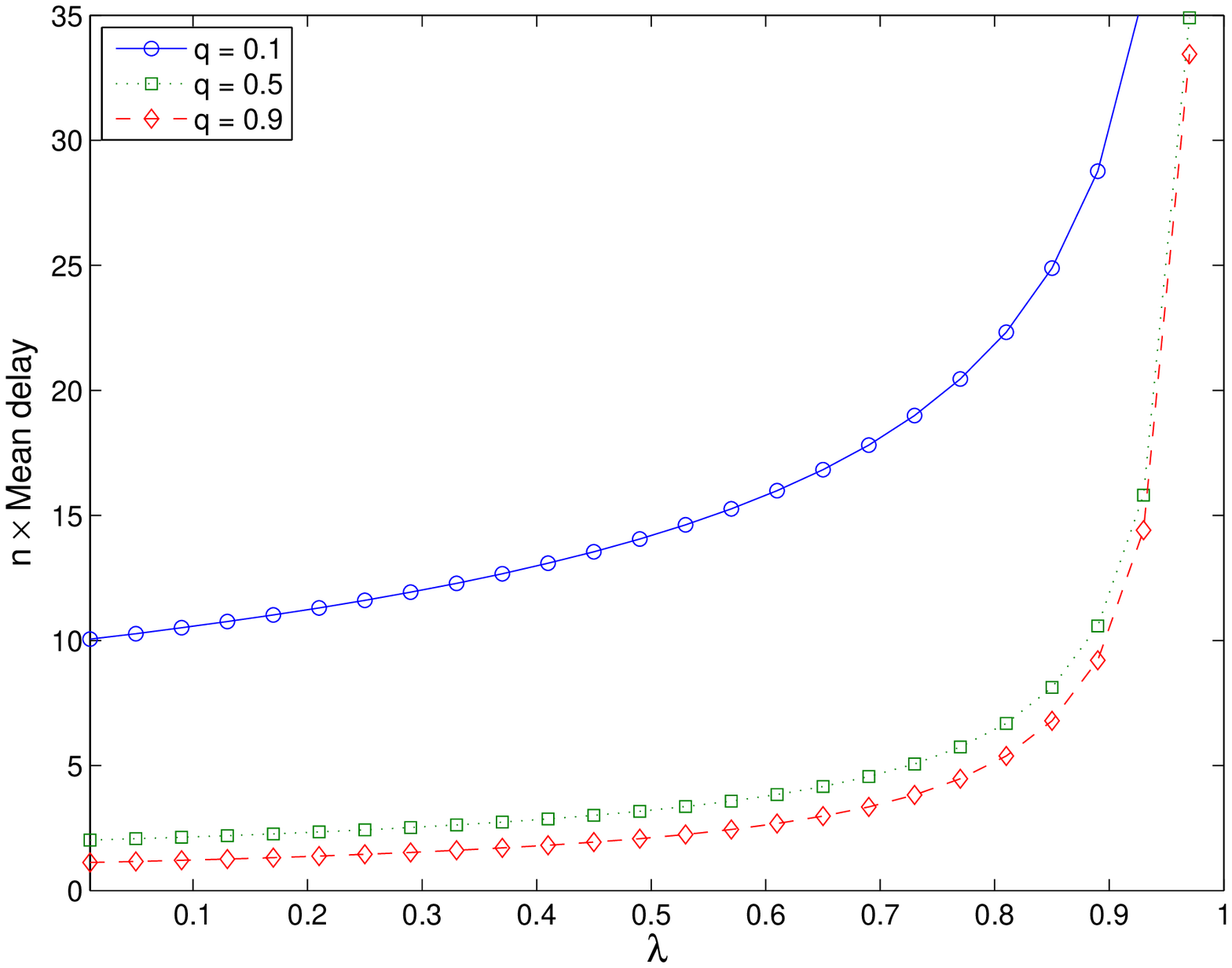}
\label{fig.delayLCQ}}\hspace*{-1cm}
\subfigure[LCQ($d$)]{\includegraphics[width=3.8in]{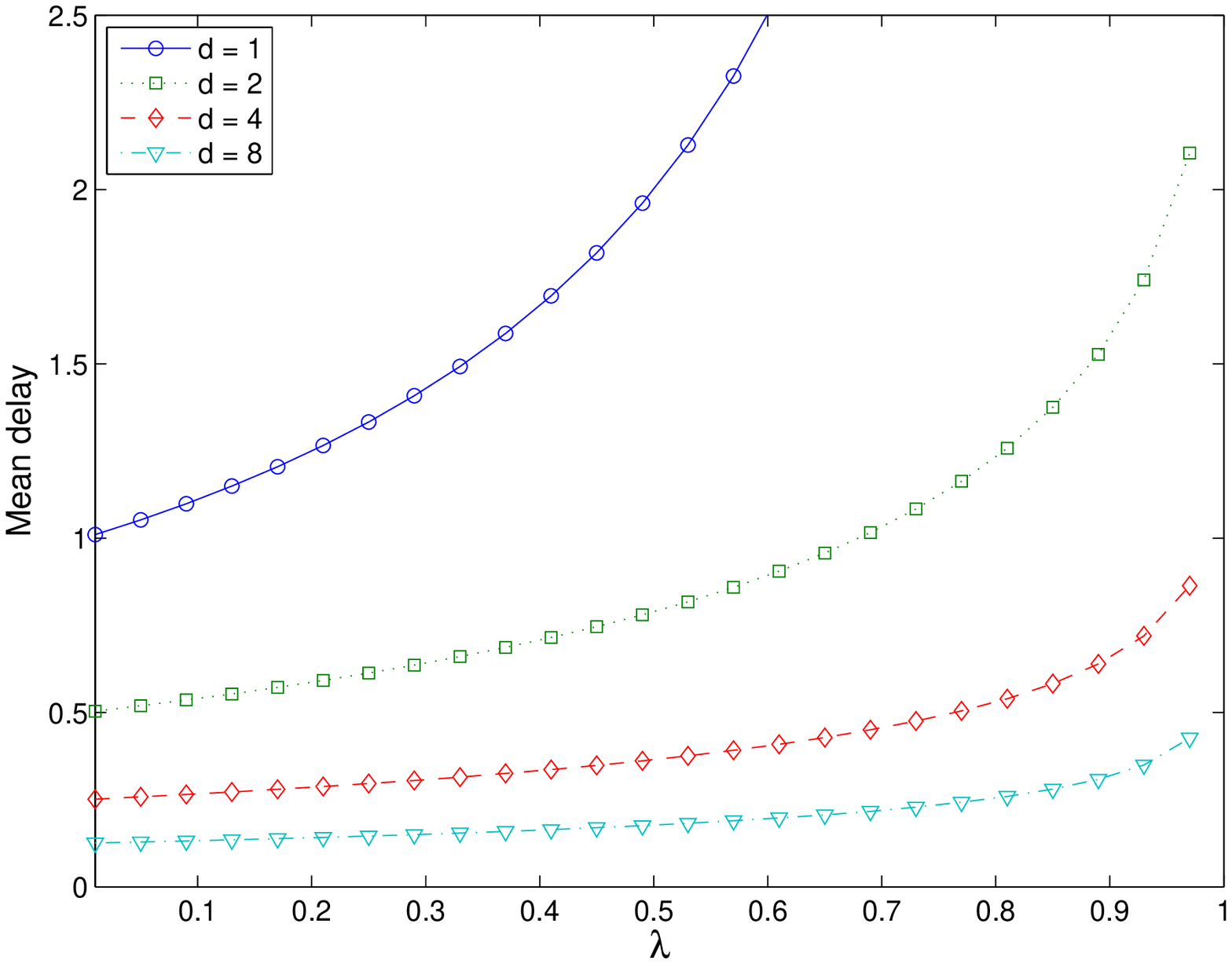}
\label{fig.2average}} \caption{Mean packet delays as
$n\rightarrow\infty$. Note the normalization in (a).}
\end{figure*}

The rest of this paper is organized as follows. We continue in
Section~\ref{sec.model} with formal description of the model and
the notation adopted in the paper. The policies LCQ$(d)$ and LCQ
are analyzed respectively in Sections~\ref{sec.LCQ-d} and
\ref{sec.LCQ}. The paper concludes with final remarks in
Section~\ref{sec.conc}.

\section{Queueing Model}\label{sec.model}

Consider $n$ queues each serving a dedicated stream of packet
arrivals as illustrated in Figure~\ref{fig.1system}. Arrivals of
each stream occur according to an independent Poisson process with
rate $\lambda<1$ packets per unit time and transmission time of
each packet is exponentially distributed with mean $n^{-1}$,
chosen independently of the prior history of the system. Each
queue is serviced by a designated channel but at most one channel
can transmit at a time. Channel states fluctuate randomly and each
channel is eligible for transmission with probability~$q\in (0,1]$
independently of other channels. Queues with eligible channels are
called {\em connected}. We assume that channel states remain
constant during packet transmission and that they are determined
anew, independently of each other and of the current queue
lengths, just before the next transmission decision.

\begin{figure}\vspace*{-0.5cm}\begin{center}\hspace*{-0cm}
\epsfig{file=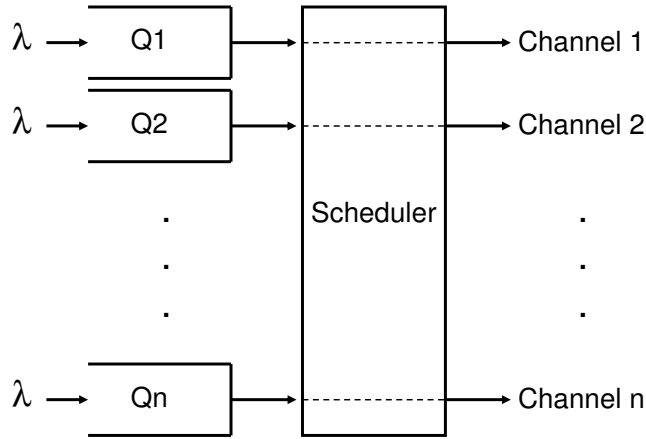, width=4in}\end{center} \vspace*{-0.8cm}
\caption{Sketch of the considered queueing system. At most one
queue is serviced at a time, and each queue~$i$ can be serviced
only when channel $i$ is eligible for transmission.}
\label{fig.1system}\end{figure}

Let $m_k(t)$ denote the number of queues with $k$ or more packets
at time~$t$, and let
\[u_k(t)=\frac{m_k(t)}{n},~~~~~k=0,1,2,\cdots.\]
be the fraction of such queues in the system. In particular
\begin{equation}
1 = u_0(t)\geq u_1(t)\geq\cdots\geq 0,\label{eq.u}
\end{equation}
the sequence $\{1-u_k(t)\}_{k=0}^\infty$ is the empirical
cumulative distribution function of queue occupancies, and
$\sum_{k\ge 1}u_k(t)$ is the empirical average queue occupancy in
the system at time~$t$. We denote by $P_n$ and $E_n$ respectively
probabilities and expectations associated with system size $n$. In
particular if $q_i(t)$ denotes the occupancy of the $i$-th queue
at time $t$ then by the symmetry of the model
\[E_n[u_k(t)]=P_n(q_i(t)\ge k).\]

Let  $U$ denote the collection of sequences
$\bu=\{u_k\}_{k=0}^\infty$ that satisfy relation~(\ref{eq.u}), and
endow $U$ with metric $ \rho$ that is defined by
\[\rho(\bu,\bu')=\sup_{k>0}\frac{|u_k-u_k'|}{k},~~~~~\bu,\bu'\in U.\]
Note that convergence in $U$ is equivalent to coordinate-wise
convergence, and that $U$ is compact as each coordinate lies in a
compact interval.

For each time $t$ let $\bu(t)$ denote the sequence
$\{u_k(t)\}_{k=0}^\infty$. We represent the trajectory
$(\bu(t):t\ge 0)$ by the symbol $\bu(\cdot)$, and say that
$\bu(\cdot)$ converges to a given trajectory $\bv(\cdot)$
uniformly on compact time-sets (uoc) if for all $t>0$
\[\sup_{0\le s\le t}\rho(\bu(s),\bv(s))\rightarrow 0~~\mbox{ a.s. as } n\rightarrow\infty.\]

\section{LCQ($d$)}\label{sec.LCQ-d}

We focus on LCQ($d$) which  randomly and independently selects $d$
connected queues and transmits from a longest queue within this
collection. For convenience of analysis we assume that repetitions
are allowed in the selection procedure, and that if all selected
queues are empty or no connected queue exists at a scheduling
instant then the scheduler makes a new selection after idling for
the transmission time of a hypothetical packet. This latter
assumption can be seen to imply that scheduling instances form a
Poisson process of rate $n$.

For $k= 1,2,3,\cdots$ let $\be_k=\{e_k(i)\}_{i=0}^\infty$ where
\[e_k(i)=1\{i=k\}.\] Here and in the rest of the paper
$1\{\cdot\}$ denotes 1 if its argument is true and 0 otherwise.
Jumps of the process $\bu(\cdot)$ are of the form $\pm
n^{-1}\be_{k}$ for some $k$. Namely $\bu(\cdot)$ changes by
$+n^{-1}\be_{k}$ whenever some queue with exactly $k-1$ packets
has a new arrival, and by $-n^{-1}\be_k$ whenever a packet
transmission is scheduled from a queue with exactly $k$ packets.
The number of queues with $k-1$ packets at time $t$ is given by
$n(u_{k-1}(t)-u_k(t))$; hence the former event occurs at
instantaneous rate $n\lambda(u_{k-1}(t)-u_k(t))$. The latter event
occurs if and only if, upon completion of a packet transmission,
$(i)$ there exists a connected queue and $(ii)$ the scheduler
inspects at least one connected queue with $k$ packets but none
with more than $k$ packets. To determine the instantaneous rate of
this event let $\tau$ be a scheduling instant and let
\[\alpha_n=1-(1-q)^n.\] Namely $\alpha_n$ is the probability that
there exists a connected queue at time $\tau$. Since channel
states are assigned independently of queue lengths at time $\tau$,
a connected queue at this time has strictly less than $k$ packets
with probability $1-u_{k}(\tau)$. Therefore, given that a
connected queue exists, the maximum queue length inspected by the
scheduler is equal to $k$ with (conditional) probability
$(1-u_{k+1}(\tau))^d-(1-u_k(\tau))^d$. Since scheduling instants
occur at constant rate $n$, instantaneous rate of transmissions
from a queue of size $k$ is
$n\alpha_n((1-u_{k+1}(t))^d-(1-u_k(t))^d)$. In particular
$\bu(\cdot)$ is a time-homogeneous Markov process whose generator
can be sketched as \beq \bu\leftarrow\left\{\begin{array}{ll}
\bu+n^{-1}\be_k&\mbox{ at rate }~~n\lambda(u_{k-1}-u_k)\\
\bu-n^{-1}\be_k&\mbox{ at rate
}~~n\alpha_n((1-u_{k+1})^d-(1-u_k)^d),
\end{array}\right.~~~~~~k=1,2\cdots.\label{eq.LcqdGenerator}
\eeq

It offers some convenience in the subsequent discussion to
represent the process $\bu(\cdot)$ via the ``random time change"
construction of~\cite[Chapter 6]{EthierKurtz}. Namely \bqa
\bu(t)&=&\bu(0)+\sum_{k=1}^{\infty}n^{-1}\be_kA_{k-1}\left(n\lambda\int_0^tu_{k-1}(s)-u_k(s)ds\right)\nonumber\\
&&~~~~~~~~~~~~-\sum_{k=1}^{\infty}n^{-1}\be_kD_{k}\left(n\alpha_n\int_0^t(1-u_{k+1}(s))^d-(1-u_k(s))^d
ds\right) \label{eq.RTCconst} \eqa where
$A_{k-1}(\cdot),~D_k(\cdot)$, $k=1,2,\cdots$, are mutually
independent Poisson processes each with unit rate. In informal
terms, the processes $A_k(\cdot)$ and $D_k(\cdot)$ clock
respectively arrivals to and departures from some queue with
length $k$, and the construction~(\ref{eq.RTCconst}) is based on
suitably expediting these processes to match the instantaneous
transition rates given in~(\ref{eq.LcqdGenerator}). Martingale
decomposition of the Poisson processes used in~(\ref{eq.RTCconst})
yields \beq
\bu(t)=\bu(0)+\sum_{k=1}^{\infty}\be_k\int_0^t\left(\lambda(u_{k-1}(s)-u_k(s))ds-\alpha_n((1-u_{k+1}(s))^d-(1-u_k(s))^d)\right)ds
+\vec{\varepsilon}(t), \label{eq.MGrepForU} \eeq where
$\vec{\varepsilon}(t)=\{\varepsilon_k(t)\}_{k=0}^\infty$ is such
that each coordinate process $\varepsilon_k(\cdot)$ is a
real-valued martingale adapted to the filtration generated by
$\bu(\cdot)$.

\begin{theorem}
Every subsequence of $\{n\}$ has a further subsequence along which
$\bu(\cdot)$ converges in distribution to a differentiable process
$\bv(\cdot)$ such that $v_0(t)\equiv 1$ and
\begin{equation}
\frac{d}{dt}v_k(t)=\lambda(v_{k-1}(t)-v_k(t))-(1-v_{k+1}(t))^d+(1-v_{k}(t))^d~~~~~~k=1,2\cdots.\label{eq.limit}
\end{equation}
\label{thm.fluid}
\end{theorem}

\begin{proof}
The sequence of processes $\bu(\cdot):n=1,2,\cdots$ is tight in
the Skorokhod space $D_U[0,\infty)$ of right continuous functions
with left limits in $U$~\cite[Chapter 3.5]{EthierKurtz}. Therefore
every subsequence has a further subsequence that converges in
distribution. By Skorokhod's Embedding Theorem~\cite[Theorem
3.1.8]{EthierKurtz} the processes can be reconstructed in an
appropriate probability space if necessary so that the convergence
occurs almost surely. Since jumps of $\bu(\cdot)$ have magnitudes
that scale with $n^{-1}$, the limit process is continuous and
convergence of $\bu(\cdot)$ can be taken uoc~\cite[Theorem
3.10.1]{EthierKurtz}. To describe a limit process $\bv(\cdot)$
note that $u_0(t)\equiv 1$, and so $\varepsilon_0(t)\equiv 0$. For
$k=1,2,\cdots$ the martingale $\varepsilon_k(\cdot)$ is square
integrable. This process has $O(n)$ jumps per unit time and each
jump is of size $n^{-1}$, hence its quadratic variation vanishes
as $n\rightarrow\infty$. In turn, Doob's $L^2$
inequality~\cite[Proposition 2.2.16]{EthierKurtz} implies that
$\varepsilon_k(\cdot)\rightarrow 0$ uoc as $n\rightarrow\infty$.
Since $\alpha_n\rightarrow 1$ and convergence of $\bu(\cdot)$ is
uoc, the $k$th integral in equality~(\ref{eq.MGrepForU}) converges
to
\[\int_0^t\left(\lambda(v_{k-1}(s)-v_k(s))ds-((1-v_{k+1}(s))^d-(1-v_k(s))^d)\right)ds.\]
Therefore $\bv(\cdot)$ satisfies equality~(\ref{eq.MGrepForU})
with $\alpha_n=1$ and $\vec{\varepsilon}(t)\equiv\vec{0}$.
Differential representation of that equality is~(\ref{eq.limit}).
\end{proof}

Let $U_o$ denote the set of system states in which average queue
occupancy is finite. That is,
\[U_o=\{\bu\in U: \sum_{k=1}^{\infty}u_k<\infty\}.\]
Let $\bv^*=\{v^*_k\}_{k=0}^{\infty}\in U$ be defined by setting
$v^*_0=1$ and \beq v^*_k=1-\sqrt[d]{1-\lambda v^*_{k-1}},~~~~~~k =
1,2,\cdots. \label{eq.vDef} \eeq Since $1-\sqrt[d]{1-\lambda
v^*_{k-1}}\le \lambda v^*_{k-1}$  it follows that
$v^*_k\le\lambda^k$; in particular $\bv^*\in U_o$. It can be
readily verified by substitution that $\bv^*$ is an equilibrium
point for the differential system~(\ref{eq.limit}). The following
lemma establishes that $\bv^*$ is the unique stable equilibrium
for trajectories that start in $U_o$.

\begin{lemma}
Let $\bv(\cdot)$ solve the differential system~(\ref{eq.limit})
with initial state $\bv(0)\in U_o$. Then \beq
\lim\limits_{t\rightarrow \infty}v_k(t)=v^*_k,~~~~~~~k=1,2,\cdots.
\label{eq.limit_v_k} \eeq \label{lemma.limit}
\end{lemma}

We provide a proof based on the following auxiliary result:

\begin{lemma} Let $\bv^+(\cdot)$ and $\bv^-(\cdot)$
solve the differential system~(\ref{eq.limit}) with respective
initial conditions $\bv^+(0),\bv^-(0)\in U$ such that $v_k^+(0)\ge
v^-_k(0)$ for all $k$. Then $v^+_k(t)\ge v^-_k(t)$ for all $k$ and
all $t>0$. \label{lemma.monotone}
\end{lemma}

\begin{proof}
Suppose that the lemma is incorrect and let $t>0$ be the first
instant such that
\[v^+_k(t)=v^-_k(t)~~ \mbox{ and }~~
\frac{d}{dt}v_k^+(t)<\frac{d}{dt}v^-_k(t) ~~~\mbox{ for some }
k.\] Let $i$ be the largest index $k$ that satisfies this
condition at time $t$. Then by (\ref{eq.limit})
\[\frac{d}{dt}v_i^+(t)-\frac{d}{dt}v^-_i(t)=
\lambda(v^+_{i-1}(t)-v^-_{i-1}(t))+(1-v^-_{i+1}(t))^d-(1-v^+_{i+1}(t))^d.\]
The right hand side of this equality is nonnegative due to the
choice of $t$ (since otherwise either the condition
$v^+_{i-1}(t)\ge v^-_{i-1}(t)$ and or the condition
$v^+_{i+1}(t)\ge v^-_{i+1}(t)$ must be violated before time $t$).
This contradicts with the definition of $t$; therefore no such $t$
exists and the lemma holds.
\end{proof}

\begin{proofl}{\bf\ref{lemma.limit}}
Let $\bv^+(\cdot)$ and $\bv^-(\cdot)$ be solutions
to~(\ref{eq.limit}) with respective initial states $\bv^+(0)$ and
$\bv^-(0)$ that are defined by setting
$v^+_k(0)=\max\{v_k(0),v^*_k\}$ and
$v^-_k(0)=\min\{v_k(0),v^*_k\}$ for $k=0,1,2,\cdots$. By
Lemma~\ref{lemma.monotone} \beq
v^-_k(t)~\le~v_k(t),v^*_k~\le~v^+_k(t),~~~\mbox{ for all }k,t.
\label{eq.sandwich} \eeq Equality (\ref{eq.limit}) and definition
(\ref{eq.vDef}) of $\bv^*$ give \bqaa
\frac{d}{dt}\sum_{i=k}^\infty v^\pm_i(t)&=&\lambda v^\pm_{k-1}(t)
+(1-v^\pm_k(t))^d-1\\&=& \lambda(v^\pm_{k-1}(t)-v^*_{k-1}) +
(1-v^\pm_k(t))^d-(1-v^*_k)^d, \eqaa or, in integral form,
\beq\sum_{i=k}^\infty v^\pm_i(t)-\sum_{i=k}^\infty v^\pm_i(0)=
\int_0^t\lambda(v^\pm_{k-1}(s)-v^*_{k-1})ds +
\int_0^t((1-v^\pm_k(s))^d-(1-v^*_k)^d)ds. \label{eq.integral} \eeq
Note that since $v^+_1(t)\ge v^*_1=1-\sqrt[d]{1-\lambda}$ it
follows that
\[\frac{d}{dt}\sum_{i=1}^\infty v^+_k(t)~=~\lambda
+(1-v^+_1(t))^d-1~\le~ 0.\] Hence $\sum_{i=k}^\infty v^+_i(t)$,
and therefore $\sum_{i=k}^\infty v^-_i(t)$, is bounded by
$\sum_{k=1}^\infty v_k^+(0)$ uniformly for all $t$. In turn
equality~(\ref{eq.integral}) yields
\[\left|~\int_0^t\lambda(v^\pm_{k-1}(s)-v^*_{k-1})ds +
\int_0^t((1-v^\pm_k(s))^d-(1-v^*_k)^d)ds~\right|~\le~\sum_{i=k}^\infty
v^+_i(0). \] The bound on the right hand side is finite since
$\bv^+(0)\in U_o$ due to the hypothesis $\bv(0)\in U_o$. Note
that, owing to the inequality~(\ref{eq.sandwich}), neither one of
the two integrands above changes sign. Hence if the first integral
converges as $t\rightarrow\infty$ then so does the second one,
implying further that \beq
\lim_{t\rightarrow\infty}v^\pm_k(t)=v^*_k. \label{eq.pmlimit} \eeq
Since $v^\pm_0(t)\equiv v^*_0=1$, this is clearly the case for
$k=1$. Induction on $k$ confirms that equality~(\ref{eq.pmlimit})
holds for all $k$. The desired conclusion (\ref{eq.limit_v_k}) now
follows from the property (\ref{eq.sandwich}).
\end{proofl}

Theorem~\ref{thm.fluid}, which establishes convergence over finite
time intervals, is complemented next by showing that equilibrium
distribution of $\bu(\cdot)$ converges as $n\rightarrow\infty$ to
the deterministic measure concentrated at~$\bv^*$.

\begin{theorem}
The process $\bu(\cdot)$ is ergodic. Let
$\bu^*=\{u^*_k\}_{k=0}^{\infty}$ denote the equilibrium random
variable. For $\varepsilon>0$ \beq
\lim_{n\rightarrow\infty}P_n(\rho(\bu^*,\bv^*)>\varepsilon)=0.
\label{eq.thm.equilibrium}\eeq In particular
$\lim_{n\rightarrow\infty}E_n[u^*_k]=v^*_k$ for $k=0,1,2,\cdots$.
\label{thm.equilibrium}
\end{theorem}

\begin{proof}
Let $U_n=\{\bu\in U: nu_k\in \mathbb{Z}_+ \mbox{ for } k\ge 0\}$.
Note that $(\bu(t):t\ge 0)$ is irreducible in $U_n$ and $U_n$ is
compact; therefore the process is ergodic and has a unique
equilibrium distribution concentrated on $U_n$. In that
equilibrium the rate of arrivals to queues with occupancy $k$ or
higher should be equal to the rate of departures from such queues.
That is,
\[E_n[\lambda u^*_{k-1}]~=~1-E_n[(1-u^*_k)^d]~\ge~E_n[u^*_k],~~~\mbox{ for }k\ge 1,\]
where the inequality follows since $(1-u^*_k)^d\le (1-u^*_k)$.
Therefore $E_n[u^*_k]\le\lambda^k$ and in turn
$E_n[\sum_{k=1}^\infty u^*_k]\le \lambda/(1-\lambda)$. Let
$U_{o,\lambda}\triangleq\{\bu\in U:\sum_{k=1}^\infty u_k\le
\lambda/(1-\lambda)\}$ so that $P_n(\bu^*\in U_{o,\lambda})=1$.

Suppose that (\ref{eq.thm.equilibrium}) is false so that for some
infinite subsequence $\{n'\}$ of $\{n\}$ and some $\delta>0$ \beq
P_{n'}(\rho(\bu^*,\bv^*)>\varepsilon)>\delta. \label{eq.contra}
\eeq Due to Lemma~\ref{lemma.limit} and compactness of
$U_{o,\lambda}$ one can choose $t(\varepsilon)$ such that if
$\bv(0)\in U_{o,\lambda}$ then $\rho(\bv(t),\bv^*)<\varepsilon/2$
for $t\ge t(\varepsilon)$. Let $\bu(0)$ have the same distribution
as $\bu^*$ and let $\bv(0)=\bu(0)$. By Theorem~\ref{thm.fluid}
there is a further subsequence $\{n''\}$ of $\{n'\}$ such that
$P_{n''}(\rho(\bu(t(\varepsilon)),\bv(t(\varepsilon)))>\varepsilon/2)<\delta$
whenever $n''$ is large enough. Since $P_{n''}(\bu^*\in
U_{o,\lambda})=1$, the choice of $t(\varepsilon)$ implies that
$P_{n''}(\rho(\bu(t(\varepsilon)),\bv^*)>\varepsilon)<\delta$ for
those values of $n''$. However $\bu(t(\varepsilon))$ and $\bu^*$
have identical distributions as the latter is in equilibrium;
leading to a contradiction with (\ref{eq.contra}). Hence no
sequence $\{n'\}$ and constant $\delta>0$ satisfy
(\ref{eq.contra}); so (\ref{eq.thm.equilibrium}) holds. By
definition of $\rho$ equality (\ref{eq.thm.equilibrium}) implies
that each entry $u^*_k$ of $\bu^*$ converges in probability to the
constant $v^*_k$; since $0\le u^*_k\le 1$, so does $E_n[u_k^*]$.
\end{proof}

We conclude the discussion of LCQ($d$) with a relationship between
$d$ and the tail probabilities of equilibrium queue occupancy:

\begin{theorem} $v^*_k=\Theta((\lambda/d)^k)$ as
$k\rightarrow\infty$. \label{thm.tail}
\end{theorem}

\begin{proof}
The assertion is immediate for $d=1$ so we consider the case
$d>1$. Equality~(\ref{eq.vDef}), together with Taylor expansion of
$\sqrt[d]{1-x}$ around $x=0$ yields \beq
v^*_k=\frac{\lambda}{d}v^*_{k-1}+\frac{1}{d}\sum_{i=2}^\infty\frac{(\lambda
v^*_{k-1})^i}{i!}\prod_{j=1}^{i-1}(j-\frac{1}{d}),~~~~~~~~~k=1,2,\cdots.
\label{eq.tail1} \eeq The second term on the right hand side is
nonnegative; therefore $v^*_k\ge (\lambda/d)v^*_k$ and \beq
\liminf\limits_{k\rightarrow\infty} \frac{v^*_k}{(\lambda/d)^k}\ge
1. \label{eq.thm.tail.aux1} \eeq We define $c_k\triangleq
v^*_k/(\lambda/d)^k$ and complete the proof by showing that
$\{c_k\}_{k=0}^{\infty}$ is uniformly bounded. Let $\beta_k$ be
defined as \beq \beta_k~\triangleq~1+\sum_{i=1}^\infty(\lambda
v^*_{k-1})^{i} \frac{2}{(i+2)!}\prod_{j=2}^{i+1}(j-\frac{1}{d})
~<~1+\sum_{i=1}^\infty(\lambda v^*_{k-1})^{i}
\label{eq.thm.tail.aux2} \eeq so that equality~(\ref{eq.tail1})
can be rearranged as \beq
v^*_k=\frac{\lambda}{d}v^*_{k-1}+\frac{(\lambda
v^*_{k-1})^2}{2d}(1-\frac{1}{d})\beta_k. \label{eq.thm.tail.aux3}
\eeq Since $v^*_k\le\lambda^k$ there exists a finite $k_o$ such
that $v_{k-1}^*<1/d$ for $k\ge k_o$. The bound
in~(\ref{eq.thm.tail.aux2}) implies that
$\beta_k<1/(1-\lambda/d)<d/(d-1)$ for such $k$; in turn by
(\ref{eq.thm.tail.aux3})
\[c_k<c_{k-1}+c_{k-1}^2\frac{d}{2}(\lambda/d)^k,~~~~~~~~~k>k_o.\] It
can be verified by induction on $k>k_o$ that \beq
c_k<c_{k_o}(1+\lambda+\lambda^2+\cdots+\lambda^{k-k_o}):
\label{eq.thm.tail.aux3.5} \eeq Namely, if
(\ref{eq.thm.tail.aux3.5}) holds for $k$ then it holds also for
$k+1$ if
$c_{k_o}(1+\lambda+\lambda^2+\cdots+\lambda^{k-k_o})^2\lambda^{k_o}/(2d^{k})<1$.
This latter condition can be verified based on the bound
$c_{k_o}(\lambda/d)^{k_o}=v_{k_o}<1/d$, which follows from the
definition of $k_o$. Inequality~(\ref{eq.thm.tail.aux3.5}) implies
the uniform bound $c_k<c_{k_o}/(1-\lambda)$; therefore \beq
\limsup\limits_{k\rightarrow\infty} \frac{v^*_k}{(\lambda/d)^k} <
\frac{c_{k_o}}{1-\lambda}. \label{eq.thm.tail.aux4} \eeq The
theorem follows due to (\ref{eq.thm.tail.aux1}) and
(\ref{eq.thm.tail.aux4}).
\end{proof}

\section{LCQ}\label{sec.LCQ}

Given $\bmm(\tau)$ at a scheduling instant $\tau$, the maximum
occupancy over all connected queues at time $\tau$ is equal to
$k=1,2,\cdots$ with probability
$(1-q)^{m_{k+1}(\tau)}-(1-q)^{m_k(\tau)}$; hence under the LCQ
policy $\bu(\cdot)$ is a time-homogenous Markov process with jump
rates \beq \bu\leftarrow\left\{\begin{array}{ll}
\bu+n^{-1}\be_k&\mbox{ at rate }~~n\lambda(u_{k-1}-u_k)\\
\bu-n^{-1}\be_k&\mbox{ at rate
}~~n\left((1-q)^{nu_{k+1}}-(1-q)^{nu_k}\right).
\end{array}\right.
\label{eq.lcq.generator} \eeq This process can be constructed as
in Section~\ref{sec.LCQ-d}, so that \beq
u_k(t)=u_k(0)+\int_0^t\left(\lambda(u_{k-1}(s)-u_k(s))-(1-q)^{m_{k+1}(s)}+(1-q)^{m_{k}(s)}\right)ds+\varepsilon_k(t)
\label{eq.u-lcq}\eeq where $\varepsilon_k(\cdot)$ is a martingale
that vanishes as $n\rightarrow\infty$. The sequence of processes
$\bu(\cdot): n=1,2,\cdots$ converges in distribution along
subsequences of $\{n\}$, but identifying a limit is relatively
more involved than for the LCQ$(d)$ policy since the process
$\bmm(\cdot)$ fluctuates persistently for all values of $n$ and
the integrand in (\ref{eq.u-lcq}) does not converge. Rather than
this integrand, here we study the behavior of the integral
in~(\ref{eq.u-lcq}) via an averaging technique due to
Kurtz~\cite{Kurtz92}. In reading this section the reader may find
it helpful to consult related applications of this technique
in~\cite{MA03,HuntKurtz94,Zachary02}.

Let $\Omega$ denote the set of sequences
$\bm=\{\omega_k\}_{k=0}^\infty$ such that $\omega_k\in
\mathbb{Z}_+\cup\{\infty\}$, $\omega_0=\infty$, and $\omega_k\ge
\omega_{k+1}$. Define the mapping $h:\Omega\mapsto [0,1]^\infty$
by setting
\[h(\bm)=\{(1+\omega_k)^{-1}\}_{k=0}^{\infty},~~~~~\bm\in\Omega,\] with the
understanding that $1+\infty=\infty$ and $1/\infty=0$. Let
$\Omega$ be endowed with metric $\rho_o$ defined by
\[\rho_o(\bm,\bm')=\rho(h(\bm),h(\bm')),~~~~~\bm,\bm'\in\Omega.\]
In particular $\Omega$ is compact with respect to the induced
topology. We denote by $\mathcal{L}$ the collection of measures
$\mu$ on the product space $[0,\infty)\times \Omega$ such that
$\mu([0,t)\times \Omega)=t$ for each $t>0$. Let $\mathcal{L}$ be
endowed with the topology corresponding to weak convergence of
measures restricted to $[0,t)\times \Omega$ for each~$t$. Since
$\Omega$ is compact, so is $\mathcal{L}$ due to Prohorov's
Theorem.

Let $\xi$ be a random member of $\mathcal{L}$ defined by
\[\xi([0,t)\times A)=\int_0^t1\{\bmm(s)\in A\}ds,~~~t>0,~A\in\mathcal{B}(\Omega).\]
Here $\mathcal{B}(\Omega)$ denotes Borel sets of $\Omega$. Note
that equality~(\ref{eq.u-lcq}) can be expressed in terms of $\xi$
as
\[u_k(t)~=~u_k(0)+\int_0^t\lambda(u_{k-1}(s)-u_k(s))ds-\phi_{k+1}(t)+\phi_{k}(t)+
\varepsilon_k(t) \] where
\[\phi_k(t)~\triangleq~\int_0^t(1-q)^{m_k(s)}ds~=~\int_{[0,t)\times
\Omega}(1-q)^{\omega_k} \xi(ds\times d\bm).\] Compactness of
$\mathcal{L}$ implies that each subsequence of $\{n\}$ has a
further subsequence along which $\xi$ converges in distribution.
This property is also possessed by $\bu(\cdot)$, and therefore by
the pair $(\bu(\cdot),\xi)$.

The following definition is useful in characterizing possible
limits of $(\bu(\cdot),\xi)$: For fixed $\bu\in U$ let
$\omega^\bu(\cdot)$ denote the Markov process with states in
$\Omega$ and with the following transition rates: \beq
\omega^\bu\leftarrow\left\{\begin{array}{ll}
\omega^\bu+\be_k&\mbox{ at rate } \lambda(u_{k-1}-u_k)\\
\omega^\bu-\be_k&\mbox{ at rate }
(1-q)^{\omega^\bu_{k+1}}-(1-q)^{\omega^\bu_k}.
\end{array}\right. \label{eq.w.generator} \eeq
See Figure~\ref{fig.Fast.lcq} for a partial illustration of this
process. The process $\omega^\bu(\cdot)$ bears a certain
resemblance to $\bmm(\cdot)=n\bu(\cdot)$, which can be observed by
inspecting the generators (\ref{eq.lcq.generator}) and
(\ref{eq.w.generator}), though it should be noted that in
(\ref{eq.w.generator}) $\bu=\{u_k\}_{k=0}^\infty$ is a constant
and has no binding to instantaneous values of $\omega^\bu(\cdot)$.
We also point out that $\omega^\bu(\cdot)$ evolves on a
compactified state space and it is reducible due to the states
that involve $\infty$; hence it has multiple equilibrium
distributions in general.
\begin{figure}
\vspace*{-1.2cm}\begin{center}\hspace*{-0.5cm}
\epsfig{file=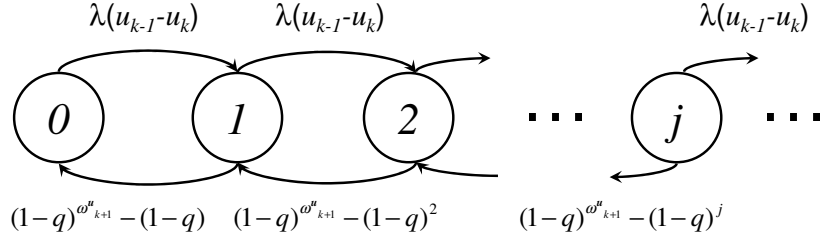, width=5in}\vspace*{-5cm}\end{center}
\caption{Transition rates of $\omega^\bu_k(\cdot)$, that is, the
$k$th coordinate of $\omega^{\bu}(\cdot)$. The process has also an
isolated state $\infty$ which is not shown. The coordinate process
$\omega^\bu_k(\cdot)$ is generally not Markovian due to its
dependence on $\omega^\bu_{k+1}(\cdot)$.} \label{fig.Fast.lcq}
\end{figure}

\begin{theorem}\label{thm.lcq.0}
Let $(\bv(\cdot),\chi)$ be the limit of $(\bu(\cdot),\xi)$ along a
convergent subsequence of $\{n\}$.

a) The limit measure $\chi$ satisfies
\[\chi([0,t)\times A) =
\int_0^t\pi_{\bv(s)}(A)ds,~~~~~~~~~t>0,~A\in\mathcal{B}(\Omega),
\] where, for each $s>0$, $\pi_{\bv(s)}$ is an equilibrium
distribution for the process $\omega^{\bv(s)}(\cdot)$ such that
\[\pi_{\bv(s)}\left(\omega_k=\infty\right)=1 ~~\mbox{ if }~~v_k(s)>0.\]

b) The limit trajectory $\bv(\cdot)$ satisfies \beq \hspace*{-0cm}
\frac{d}{dt}v_k(t)=\lambda(v_{k-1}(t)-v_k(t)) -
E_{\pi_{\bv(t)}}[(1-q)^{\omega_{k+1}}-(1-q)^{\omega_k}],
\label{eq.lcq.1} \eeq where $k=1,2,\cdots$ and $E_{\pi_{\bv(t)}}$
denotes expectation with respect to distribution $\pi_{\bv(t)}$.
\end{theorem}

\begin{proof}
Let all processes be constructed on a common probability space so
that convergence of $(\bu(\cdot),\xi)$ is almost sure. Convergence
of $\bu(\cdot)$ is then uoc.  We start by consulting~\cite[Lemma
1.4]{Kurtz92} to verify that the limit measure $\chi$ possesses a
density so that \beq \chi([0,t)\times A) =
\int_0^t\gamma_s(A)ds,~~~~~~~~~t>0,~A\in\mathcal{B}(\Omega),
\label{eq.thm.lcq.0.aux0} \eeq where, for each $s$, $\gamma_s$ is
a probability distribution on $\Omega$. We proceed by identifying
these distributions.

Let $\mathcal{F}$ denote the collection of bounded continuous
functions $f:\Omega\mapsto \mathbb{R}$ such that $f(\bm)$ depends
on a finite number of entries in the sequence
$\bm=\{\omega_k\}_{k=0}^\infty\in\Omega$. Given $f\in \mathcal{F}$
define the function $G^f:\Omega\times U\mapsto \mathbb{R}$ by
setting
\[G^f(\bm,\bu)\triangleq\sum_{k=1}^{\infty}(f(\bm+\be_k)-f(\bm))
\lambda(u_{k-1}-u_k)+(f(\bm-\be_k)-f(\bm))((1-q)^{\omega_{k+1}}-(1-q)^{\omega_k})\]
for each $\bm\in\Omega$ and $\bu=\{u_k\}_{k=0}^\infty\in U$. $G^f$
is continuous due to the continuity of $f$, and continuity of
$G^f$ is uniform since the product space $\Omega\times U$ is
compact.

The process $f(\bmm(\cdot))$ satisfies at each instant $t$
\begin{eqnarray*}&&\!\!\!\!\!\!\!\!\!\!\!f(\bmm(t))-f(\bmm(0))\\
&&~~~~~~~~~=~~\int_0^t\sum_{k=1}^{\infty}
(f(\bmm(s)+\be_k)-f(\bmm(s)))
dA_{k-1}\left(\int_0^sn\lambda(u_{k-1}(\tau)-u_k(\tau))d\tau\right) \\
&&~~~~~~~~~~~~~+
\int_0^t\sum_{k=1}^{\infty}(f(\bmm(s)-\be_k)-f(\bmm(s)))
dD_k\left(\int_0^sn\left((1-q)^{m_{k+1}(\tau)}-(1-q)^{m_k(\tau)}\right)d\tau\right) \\
&&~~~~~~~~~=~~n\int_0^t\sum_{k=1}^{\infty}(f(\bmm(s)+\be_k)-f(\bmm(s)))
\lambda(u_{k-1}(s)-u_k(s))ds\\
&&~~~~~~~~~~~~~+n\int_0^t\sum_{k=1}^{\infty}(f(\bmm(s)-\be_k)-f(\bmm(s)))
\left((1-q)^{m_{k+1}(s)}-(1-q)^{m_k(s)}\right)ds+\mu^f(t)\\
&&~~~~~~~~~=~~n\int_0^t G^f(\bmm(s),\bu(s))ds+\mu^f(t),
\end{eqnarray*}
where $\mu^f(\cdot)$ is a square-integrable martingale.
Rearranging the last equality and expressing the integral there in
terms of the random measure $\xi$ yields
\[\int_{[0,t)\times\Omega}G^f(\bm,\bu(s))\xi(ds\times
d\bm) ~~=~~\int_0^t G^f(\bmm(s),\bu(s))ds~~=~~
(f(\bmm(t))-f(\bmm(0))/n+\mu^f(t)/n.\] Since $f$ is bounded, the
first term on the right hand side vanishes as
$n\rightarrow\infty$. The martingale $\mu^f(\cdot)$ has bounded
jumps; in turn by Doob's $L^2$ inequality $\mu^f(t)/n$ also
vanishes. Therefore \beq
\int_{[0,t)\times\Omega}G^f(\bm,\bu(s))\xi(ds\times
d\bm)\rightarrow 0.\label{eq.thm.lcq.0.aux1} \eeq Since
$\bu(\cdot)$ converges uoc to $\bv(\cdot)$ by hypothesis, uniform
continuity of $G^f$ implies \beq \left|
\int_{[0,t)\times\Omega}G^f(\bm,\bu(s))\xi(ds\times d\bm) -
\int_{[0,t)\times \Omega}G^f(\bm,\bv(s))\xi(ds\times
d\bm)\right|~\rightarrow~0. \label{eq.thm.lcq.0.aux2} \eeq Finally
by the Continuous Mapping Theorem \beq \int_{[0,t)\times
\Omega}G^f(\bm,\bv(s))\xi(ds\times
d\bm)\rightarrow\int_{[0,t)\times
\Omega}G^f(\bm,\bv(s))\chi(ds\times d\bm).
\label{eq.thm.lcq.0.aux3} \eeq
Observations~(\ref{eq.thm.lcq.0.aux1})--(\ref{eq.thm.lcq.0.aux3})
lead to \[ \int_{[0,t)\times \Omega}G^f(\bm,\bv(s))\chi(ds\times
d\bm)~=~\int_0^t\sum_{\bm\in\Omega}
G^f(\bm,\bv(s))\gamma_s(\bm)ds~=~0, \] where the left equality is
due to~(\ref{eq.thm.lcq.0.aux0}). This equality holds for all
$t>0$; therefore
\[\sum_{\bm\in\Omega} G^f(\bm,\bv(s))\gamma_s(\bm)=0\]
for almost all $s>0$. Since $f\in \mathcal{F}$ is arbitrary (note
that $\mathcal{F}$ is dense in continuous bounded functions on
$\Omega$) \cite[Proposition 4.9.2]{EthierKurtz} implies that
$\gamma_s$ is an equilibrium distribution for the process
$\bm^{\bv(s)}(\cdot)$.

Let $\varepsilon>0$ and $[t_0,t_1]$ be an interval such that
$v_k(t)\ge\varepsilon$ for $t\in[t_0,t_1]$. Since $v_k(t)$ is the
limit of $u_k(t)=n^{-1}m_k(t)$, for any given integer $B$
\[\xi([t_0,t_1]\times\{0,1,2,\cdots,B\})~=~\int_{t_0}^{t_1}1\{m_k(s)\le B\}ds~\rightarrow~0~~~~~\mbox{ as } n\rightarrow\infty.\]
Hence $\chi([t_0,t_1]\times \mathbb{Z}_+)=0$ due to the
arbitrariness of $B$. Since $\varepsilon$ can be chosen
arbitrarily small it follows that $\gamma_t(\mathbb{Z}_+)=0$ for
almost all $t$ such that $v_k(t)>0$. This completes the proof of
part a). Part b) follows from equality~(\ref{eq.u-lcq}) since
\[\int_0^t(u_{k-1}(s)-u_k(s))ds\rightarrow \int_0^t(v_{k-1}(s)-v_k(s))ds\]
due to uoc convergence of $\bu(\cdot)$, and
\[\phi_k(t)~\rightarrow~\int_{[0,t)\times
\Omega}(1-q)^{\omega_k} \chi(ds\times
d\bm)~=~\int_0^tE_{\gamma_s}[(1-q)^{w_k}]ds\] due to the
Continuous Mapping Theorem.
\end{proof}

Theorem~\ref{thm.lcq.0} explains the extent of the disparity
between time scales of two processes, namely $\bmm(\cdot)$ and its
normalized version $\bu(\cdot)$: The process $\bmm(\cdot)$
displays far larger variation than its normalized version, so
that, in the limit of large~$n$, $\bmm(\cdot)$ settles to
equilibrium before $\bu(\cdot)$ changes its value. In particular
integral of a binary-valued measurable function of $\bmm(\cdot)$
is well-approximated by integrating an appropriate equilibrium
probability. Provided that $\bu(t)$ remains close to $\bv(t)$, the
process $\bm^{\bv(t)}(\cdot)$ mimics a slowed-down version of
$\bmm(\cdot)$ observed around time~$t$; hence the alluded
equilibrium distribution pertains to $\bm^{\bv(t)}(\cdot)$.

Specification of $\bm^{\bv(t)}(\cdot)$ requires inclusion of
$\infty$ since entries of $\bmm(\cdot)$ can be as large as $n$.
\mbox{Compactifying} the augmented state-space $\Omega$ of
$\bmm(\cdot)$ via choice of the metric $\rho_o$ leads to the
representation~(\ref{eq.lcq.1}) of a limit trajectory
$\bv(\cdot)$, but it also entails ambiguity in that
representation. Namely, Theorem~\ref{thm.lcq.0} does not specify
which equilibrium distribution for $\omega^{\bv(t)}(\cdot)$ should
be adopted in~(\ref{eq.lcq.1}). While a full account of
equilibrium distributions of $\omega^{\bv(t)}(\cdot)$ appears
difficult, an important feature of the right distribution can be
identified:

\begin{lemma}\label{lemma.aux-lcq}
Let $\bv(\cdot)$ and $\pi_{\bv(\cdot)}$ be as specified by
Theorem~\ref{thm.lcq.0}. Given $k=1,2,\cdots$
\[\pi_{\bv(t)}\left(\omega_{k}\in \mathbb{Z}_+  \mbox{ and }
\omega_{k+1}=0\right)=1 \] for almost all $t$ such that
$v_k(t)=0$.
\end{lemma}

Lemma~\ref{lemma.aux-lcq} will be instrumental in obtaining a
sharper description for $\bv(\cdot)$, yet an informal explanation
may still be useful in putting it in perspective with the queueing
system of interest. Note that if $v_k(t)=0$ and $v_{k-1}(t)>0$
then $\bv(t)$ reflects a distribution with support
$\{0,1,\cdots,k-1\}$. This property does not immediately translate
into a bound on the maximum queue length in the system, since
$\bv(t)$ is the limit of $\bu(t)=n^{-1}\bmm(t)$ and so the number
of queues with at least $i\ge k$ packets, $m_i(t)$, is $o(n)$ as
$n\rightarrow\infty$. By way of interpreting $\bm^{\bv(t)}(\cdot)$
as a proxy to $\bmm(\cdot)$ around time $t$,
Lemma~\ref{lemma.aux-lcq} indicates that the maximum queue size is
at most one larger than what is deduced from $\bv(t)$ and that the
number of maximal queues is $O(1)$ as $n\rightarrow\infty$.

\begin{proofl}{\bf\ref{lemma.aux-lcq}}
Let $[t_0,t_1]$ be an interval such that $v_k(t)=0$ for
$t\in[t_0,t_1]$. We prove the lemma by showing that as
$n\rightarrow\infty$ along the convergent subsequence of interest
\bqa \xi([t_0,t_1]\times \{\bm:w_{k+1}=0\})
~=~\int_{t_0}^{t_1}1\{m_{k+1}(t)\ge 1\}dt &\rightarrow& 0,
\label{eq.lemma.lcq0.claim1} \\
\xi([t_0,t_1]\times \{\bm:w_k\in Z_+\})
~=~\int_{t_0}^{t_1}1\{m_{k}(t)\in Z_+\}dt &\rightarrow& t_1-t_0.
\label{eq.lemma.lcq0.claim2} \eqa For each integer $l$ and time
$t$ let $s_{l}(t)\triangleq\sum_{i=l}^\infty m_i(t)$. This
quantity increases when some queue with size at least $l-1$
receives a packet, and it decreases when transmission is scheduled
from some queue with size at least $l$. Given $\bu(t)$, these
events occur at respective instantaneous rates $n\lambda
u_{l-1}(t)$ and $n\left(1-(1-q)^{m_{l}(t)}\right)$. Therefore \beq
E_n[s_{l}(t_1)-s_{l}(t_0)]~=~nE_n\left[\int_{t_0}^{t_1}\lambda
u_{l-1}(t)-\left(1-(1-q)^{m_{l}(t)}\right)dt\right].
\label{eq.lemma.lcq0.newaux1} \eeq Consider this equality for
$l=k+1$. By choice of the interval $[t_0,t_1]$ \[
n^{-1}E_n[s_{k+1}(t)]~\rightarrow~\sum_{i=k+1}^\infty v_i(t)~=~0
\] and $u_k(t)\rightarrow v_k(t)=0$ for all $t\in[t_0,t_1]$.
Consequently 
\[E_n\left[\int_{t_0}^{t_1}\left(1-(1-q)^{m_{k+1}(t)}\right)dt\right]
\rightarrow 0.\] This leads to~(\ref{eq.lemma.lcq0.claim1}) since
\[ 1-(1-q)^{m_{k+1}(t)}~\ge~ q1\{m_{k+1}(t)\ge 1\}.\]
To complete the proof, note that $n^{-1}E_n[s_{k}(t)]\rightarrow
0$ for all $t\in [t_0,t_1]$; therefore
(\ref{eq.lemma.lcq0.newaux1}) evaluated at $l=k$ implies that for
any open subset $B\subset [t_0,t_1]$
\[\limsup_{n\rightarrow\infty}E_n\left[\int_{B}\left(1-(1-q)^{m_{k}(t)}\right)dt\right]
~=~\limsup_{n\rightarrow\infty}E_n\left[\int_{B}\lambda
u_{k-1}(t)dt\right]~<~\int_{B}dt. \] The last inequality is strict
since $\lambda<1$. Arbitrariness of $B$
implies~(\ref{eq.lemma.lcq0.claim2}).
\end{proofl}

Given positive integer $K$ let $U_K=\{\bu\in U: u_k=0\mbox{ for }
k\ge K\}$.

\begin{theorem}\label{thm.lcq.1}
Let $\bv(\cdot)$ and $\pi_{\bv(\cdot)}$ be as specified by
Theorem~\ref{thm.lcq.0} with initial state $\bv(0)\in U_K$ for
some $K$. Then for $t>0$

a) $\bv(t)\in U_K$ and \[ \pi_{\bv(t)}\left(\omega_{K(t)}\in
\mathbb{Z}_+ \mbox{ and } \omega_{K(t)+1}=0\right)=1
\label{eq.thm.lcq.1.claim1} \] where $K(t)=\min\{k:v_j(t)=0 \mbox{
for } j\ge k\}$.

b) \[
\frac{d}{dt}v_k(t)=\left\{\begin{array}{ll}
\lambda v_{k-1}(t)-1<0&\mbox{ if } k=K(t)-1\\
0&\mbox{ if } k\ge K(t). \end{array}\right.
\label{eq.thm.lcq.1.claim2} \] In particular $v_k(t)=0$ for
$k>0$ and $t>K/(1-\lambda)$.
\end{theorem}

\begin{proof} Let $t$ be an instant such that $K(t)<\infty$.
Lemma~\ref{lemma.aux-lcq} implies that \beq
\pi_{\bv(t)}\left(\omega_{K(t)}\in \mathbb{Z}_+ \mbox{ and }
\omega_{K(t)+1}=0\right)=1. \label{eq.thm.lcq.1.aux1} \eeq In
particular the coordinate process $\bm_{K(t)}^{\bv(t)}(\cdot)$
possesses an equilibrium in $\mathbb{Z}_+$. The process should
have equal rates of up-jumps and down-jumps in that equilibrium,
namely \beq E_{\pi_{\bv(t)}}[(1-q)^{\omega_{K(t)}}]=1-\lambda
v_{K(t)-1}(t). \label{eq.thm.lcq.1.aux2} \eeq Since
$v_{K(t)-1}(t)>0$ by definition of $K(t)$,
Theorem~\ref{thm.lcq.0}.a implies that \beq
E_{\pi_{\bv(t)}}[(1-q)^{\omega_{K(t)-1}}]=0.
\label{eq.thm.lcq.1.aux3} \eeq Substituting
(\ref{eq.thm.lcq.1.aux2}) and (\ref{eq.thm.lcq.1.aux3}) in
equality~(\ref{eq.lcq.1}) evaluated at $k=K(t)-1$ yields \beq
\frac{d}{dt}v_{K(t)-1}(t)~=~\lambda v_{K(t)-2}(t)-1~<~0.
\label{eq.thm.lcq.1.aux4} \eeq Note also that
$E_{\pi_{\bv(t)}}[(1-q)^{\omega_{K(t)+1}}]=1$ due to
(\ref{eq.thm.lcq.1.aux1}); hence equality~(\ref{eq.lcq.1}) for
$k=K(t)$ gives \beq \frac{d}{dt}v_{K(t)}(t)~=~-\lambda
v_{K(t)}(t)~=~0. \label{eq.thm.lcq.1.aux5} \eeq Since
$K(0)=K<\infty$ by hypothesis, it follows via
(\ref{eq.thm.lcq.1.aux4}) and (\ref{eq.thm.lcq.1.aux5}) that
$K(t)$ is finite and nonincreasing in $t$. Part (a) of the theorem
now follows by (\ref{eq.thm.lcq.1.aux1}). Part (b) is due to
(\ref{eq.thm.lcq.1.aux4}) and (\ref{eq.thm.lcq.1.aux5}).
\end{proof}

\begin{cor}\label{cor.lcq.2}
If $\bu(0)\in U_K$ for some $K$ then \beq
\lim_{n\rightarrow\infty}P_n(m_1(t)\in \mathbb{Z}_+,m_2(t)=0)=1
\label{eq.cor.lcq.2} \eeq for $t\ge K/(1-\lambda)$. The system
occupancy $\sum_{k=1}^\infty m_k(t)$ converges in distribution to
the equilibrium value of a birth-death process with constant birth
rate $\lambda$ and death rate $1-(1-q)^j$ at state $j$.
\end{cor}

\begin{proof}
Let $\{n_i\}$ be a subsequence along which $(\bu(\cdot),\xi)$
converges and let $(\bv(\cdot),\chi)$ denote the limit. Since
$\bu(0)\in U_K$ it follows that $\bv(0)\in U_K$. Choose
$t_1>t_0>K/(1-\lambda)$ so that by Theorem~\ref{thm.lcq.1}.b
$\bv(t)=\{1,0,0,0,\cdots\}$ for $t\in [t_0,t_1]$. Let
$A=\{\bm\in\Omega:\omega_1\in \mathbb{Z}_+, w_2=0\}$. Then
\[\int_{t_0}^{t_1} P_{n_i}(\bmm(t)\in A)dt
~=~E_{n_i}\left[\int_{t_0}^{t_1} 1\{\bmm(t)\in A\}dt \right]
~\rightarrow~ \int_{t_0}^{t_1}\pi_{\bv(t)}(A)dt~=~t_1-t_0,
\] where the last equality is due to
Theorem~\ref{thm.lcq.1}.a. The above limit does not depend on the
particular subsequence $\{n_i\}$; therefore (\ref{eq.cor.lcq.2})
follows. The final claim of the corollary is verified by observing
that for $t>K/(1-\lambda)$ the coordinate process
$\omega_2^{\bv(t)}\equiv 0$ in equilibrium; and in turn
$\omega_1^{\bv(t)}$ is a positive recurrent birth-death process on
$\mathbb{Z}_+$ with birth rate $\lambda$ and death rate
$1-(1-q)^j$ at state $j$.
\end{proof}

It should be noted that the hypothesis $\bu(0)\in U_K$ is
necessary for the conclusions of Corollary~\ref{cor.lcq.2}: If the
initial size of a single queue is allowed to grow without bound
with increasing $n$ then, for large values of $n$, that queue
receives service whenever it is connected. In effect this reduces
the service rate available to the rest of the system by a factor
of $(1-q)$. In such degenerate cases the present analysis applies
to the subsystem that is composed of queues with bounded initial
occupancies, after appropriate adjustment of the service rate.

\section{Final remarks: LCQ($d_n$)}\label{sec.conc}

Conclusions of Sections~\ref{sec.LCQ-d} and \ref{sec.LCQ} reveal
that the system occupancies under LCQ($d$) and LCQ differ by a
factor of order $n$ as $n\rightarrow\infty$. More insight on this
disparity, especially for moderate values of $d$ relative to $n$,
can be gained by considering an asymptotic regime in which $d$ is
allowed to depend on~$n$. Here we sketch asymptotic analysis of
LCQ($d_n$) in the case
\[\lim_{n\rightarrow\infty} d_n=\infty \mbox{ ~~~and~~~ }
\lim_{n\rightarrow\infty} \frac{d_n}{n}=0.\] The present
discussion closely follows that of Section~\ref{sec.LCQ}, hence
proofs are omitted.

Under LCQ($d_n$) the representation~(\ref{eq.MGrepForU}) can be
expressed as
\[u_k(t)=u_k(0)+\int_0^t\left(\lambda(u_{k-1}(s)-u_k(s))-
\left((1-\frac{b_{k+1}(s)}{d_n})^{d_n}-(1-\frac{b_k(s)}{d_n})^{d_n}\right)\right)ds
+\varepsilon_k(t)\] where $b_k(t)\triangleq d_nu_{k}(t)$. Let
$\bb(t)=\{b_k(t)\}_{k=0}^\infty$ and let $\Omega_o$ be obtained by
augmenting $\Omega$ with nonincreasing sequences that take values
in $\mathbb{R}_+\cup\{\infty\}$. Define the random measure $\xi_o$
by
\[\xi_o([0,t)\times A)=\int_0^t1\{\bb(s)\in A\}ds,~~~~~~~t>0,~A\in \mathcal{B}(\Omega_o).\]
Consideration of the pair $(\bu(\cdot),\xi_o)$ via an analogue of
Theorem~\ref{thm.lcq.0} identifies possible limits $\bv(\cdot)$ of
$\bu(\cdot)$ as solutions to
\[ \frac{d}{dt}v_k(t)=\lambda(v_{k-1}(t)-v_k(t)) -
E_{\pi_{\bv(t)}}[e^{-\omega_{k+1}}-e^{-\omega_{k}}],~~~~~~~~k=1,2,\cdots\]
where $\pi_{\bv(t)}$ is a distribution on $\Omega_o$ such that
$\pi_{\bv(t)}(\omega_k=\infty)=1$ if $v_{k}(t)>0$ and
\[E_{\pi_{\bv(t)}}\left[1\{\omega_k\ne\infty\}\left(\lambda(v_{k-1}(t)-v_k(t))+e^{-\omega_{k}}-e^{-\omega_{k+1}}\right)\right]=0.\]
The line of reasoning employed in establishing
Lemma~\ref{lemma.aux-lcq} and Theorem~\ref{thm.lcq.1} readily
applies to $\bv(\cdot)$ and $\pi_{\bv(\cdot)}$ here, yielding that
\[\pi_{\bv(t)}\left(\omega_{k}\in \mathbb{R}_+  \mbox{ and }
\omega_{k+1}=0\right)=1~~~\mbox{ if } v_k(t)=0, \] and that
$v_1(t)=0$ for $t>K(0)/(1-\lambda)$. In turn for such $t$,
$b_1(t)=O(1)$ and $b_2(t)=o(1)$ as $n\rightarrow\infty$. The
maximum queue size in equilibrium therefore tends to one, but the
number of queues at that occupancy is substantially larger than
the same number under the LCQ policy. In particular for large
enough values of $t$ the total system occupancy $\sum_{k\ge 1}
m_k(t)=(n/d_n)\sum_{k\ge 1}b_k(t)$ is $O(n/d_n)$.

\end{document}